\documentclass[12pt]{article}
\usepackage{a4wide}

\usepackage{graphicx} 
\usepackage{authblk}

\usepackage[breaklinks]{hyperref}
\usepackage{booktabs}

\usepackage{tikz}
\usetikzlibrary{matrix,arrows}

\usepackage{amsmath,amssymb,amsfonts,amsthm}
\usepackage{bbding}
\usepackage{extarrows}

\usepackage{mathrsfs,latexsym}
\usepackage[mathscr]{eucal}
\usepackage{mathrsfs}

\newtheorem{theorem}{\bf Theorem}[section]
\newtheorem{definition}[theorem]{\bf Definition}

\newcommand{\II}{{\boldsymbol{1}}}

\newcommand{\CC}{{\mathbb C}}

\newcommand{\RR}{{\mathbb R}}

\newcommand{\ZZ}{{\mathbb Z}}

\newcommand{\CoinX}[1]{C_0^\infty({#1})}

\newcommand{\DD}{{\mathscr D}}

\newcommand{\HH}{{\mathscr H}}

\newcommand{\LL}{{\mathcal L}}

\newcommand{\Ac}{{\mathcal A}}
\newcommand{\Bc}{{\mathcal B}}

\newcommand{\gb}{{\boldsymbol{g}}}
\newcommand{\hb}{{\boldsymbol{h}}}
\newcommand{\jb}{{\boldsymbol{j}}}

\newcommand{\ogth}{{\mathfrak o}}
\newcommand{\tgth}{{\mathfrak t}}

\newcommand{\supp}{{\rm supp}\,}

\newcommand{\dvol}{d\textrm{vol}}

\newcommand{\ip}[2]{{\langle #1\mid #2\rangle}}

\newcommand{\WF}{{\rm WF}\,}

\newcommand{\Ob}{{\boldsymbol{0}}}

\newcommand{\Jb}{{\boldsymbol{J}}}
\newcommand{\Lb}{{\boldsymbol{L}}}
\newcommand{\Mb}{{\boldsymbol{M}}}
\newcommand{\Nb}{{\boldsymbol{N}}}

\newcommand{\Mc}{{\mathcal{M}}}
\newcommand{\Nc}{{\mathcal{N}}}
\newcommand{\Sc}{{\mathcal{S}}}

\newcommand{\LCT}{{\sf LCT}}

\newcommand{\Loc}{{\sf Loc}}

\newcommand{\LocSrc}{{\sf LocSrc}}

\newcommand{\BG}{{\sf BkGrnd}}

\newcommand{\Sympl}{{\sf Sympl}}

\newcommand{\Alg}{{\sf Alg}}
\newcommand{\CAlg}{{\sf C^*\hbox{-}Alg}}

\newcommand{\Phys}{{\sf Phys}}

\newcommand{\AlgSts}{{\sf AlgSts}}

\newcommand{\Af}{{\mathscr A}}

\newcommand{\Bf}{{\mathscr B}}

\newcommand{\Tc}{{\mathcal T}}

\newcommand{\Zf}{{\mathscr Z}}

\newcommand{\id}{{\rm id}}

\newcommand{\nto}{\stackrel{.}{\to}}

\newcommand{\rce}{{\rm rce}}

\newcommand{\dyn}{{\rm dyn}}
\newcommand{\kin}{{\rm kin}}

\begin{document}

\title{Locally covariant quantum field theory and \\ the problem of  
formulating the same physics \\ in all spacetimes}

\author{Christopher J. Fewster\footnote{Email: {\tt chris.fewster@york.ac.uk}}}

\affil{Department of Mathematics, University of York, Heslington, York, YO10 5DD, U.K.}
%
%


\maketitle
\begin{abstract}
The framework of locally covariant quantum field theory is
discussed, motivated in part using `ignorance principles'. 
It is shown how theories can be represented by suitable
functors, so that physical equivalence of theories may be expressed via
natural isomorphisms between the corresponding functors. 
The inhomogeneous scalar field is used to illustrate the ideas. It is argued that there are two reasonable definitions of the local physical content associated with a locally covariant theory;
when these coincide, the theory is said to be dynamically local.
The status of the dynamical locality condition is reviewed, as are its
applications in relation to (a) the foundational question of what it means
for a theory to represent the same physics in different spacetimes,
and (b) a no-go result on the existence of natural states. 
\end{abstract}


\section{Introduction}

Quantum field theory (QFT) was originally developed as a theory of particle physics in Minkowski space,
in which the Poincar\'e symmetry group plays a key role.
It appears in the practical computations of Lagrangian QFT, with pervasive use of momentum space techniques, in the classification of particle species via representation theory,  and also in axiomatic approaches to the subject in which a unitary Hilbert space representation of the Poincar\'e
group and an invariant vacuum vector take centre stage~\cite{StreaterWightman, Haag}. 

Our universe, however, is not Minkowski space, but instead is well-described by a curved spacetime;
accordingly, much work has been devoted to the extension of QFT to such backgrounds. Even where the starting point is a classical Lagrangian,
for which minimal coupling can suggest a natural extension to curved spacetime, the formulation of the quantum field theory raises numerous conceptual issues 
which have gradually been solved
over the past 40 years (see e.g.,~\cite{Wald_qft}).


In the process, reliance on
spacetime symmetries, a preferred vacuum state, and even the notion of particles
have all had to be jettisoned. 
The axiomatic setting faces an additional problem. When one strips out the 
axioms of Wightman--G{\aa}rding~\cite{StreaterWightman} or Haag--Kastler--Araki~\cite{Haag}
QFT that relate to the Poincar\'e group, one is left with a rather meagre residue.
Moreover, the aim of QFT in curved spacetime is to permit the formulation of, in 
some sense, the `same' physical theory in arbitrary (sufficiently well-behaved) spacetime backgrounds.
What axiom can be given to capture this idea of formulating the same physics in 
all spacetimes (a phrase we will abbreviate as SPASs)?

Perhaps for this reason, the axiomatic development of QFT in curved spacetimes has
been comparatively underdeveloped, both in contrast to axiomatic approaches in 
Minkowski space and to the investigation of concrete QFT models in curved spacetime. 
The main purpose of this contribution is to present the axiomatic framework 
of \emph{local covariance} for
physical theories on general spacetime backgrounds, introduced by Brunetti, Fredenhagen
and Verch~\cite{BrFrVe03}, and to describe the extent to which it addresses the problem 
of SPASs~\cite{FewVer:dynloc_theory}. The ideas will be illustrated using
the inhomogeneous scalar field model, following the recent treatment~\cite{FewSch:2014}. 
There are two side themes: first, 
the absence of any viable notion of natural state compatible with local covariance to
replace the Minkowski vacuum state -- a general model-independent result proved in~\cite{FewVer:dynloc_theory} will
be described, but in addition a new and self-contained
argument will be given for the inhomogeneous scalar field; 
second, I will attempt to motivate some of the ideas presented as a 
constructive use of `ignorance principles'.


\section{Ignorance principles} 

The task of science is to reduce the complexity of the real world to basic ideas and principles from which progressively more detailed models of reality can be built.  The success of this endeavour
is all the more remarkable, given that there are many things about the world that we \emph{do} not know, and moreover many things that we \emph{cannot} know. Of course, science aims to remedy contingent ignorance, but even here progress has been greatly assisted by a fortuitous separation of scales in the structure of matter, 
permitting the development of fluid dynamics, for example, without the need to understand atoms, and of chemistry without the
need to understand quarks. 

On the other hand, it is much less obvious that science can proceed at all
in a world where there are things that cannot be known. Imagine a world in which influences of which we had no control or knowledge were at work, permeating physical phenomena on all scales and without restriction. That world would appear capricious, perhaps not even displaying statistical regularity despite the best efforts of experimentalists. A pre-requisite for the success of science, therefore, would appear to be the principle:
{\em Anything that we cannot know may be neglected}, which we  dignify with the title of the 
{\em ignorance meta-principle}. 

The unavoidable ignorance to be discussed here is imposed by the bounded
speed of propagation of all influences and signals. To the best of our knowledge, this is a law of nature, which absolutely denies us direct access to regions of spacetime that are
spacelike-separated from our own. Application of the ignorance meta-principle
leads to familiar principles of locality: that observations in a spacetime region $O$ 
should be independent of any made in the causal complement of $O$, and that the
equations of motion should be consistent with the finite speed of light. 
Each can be regarded as a consistency mechanism in the theory that maintains the
ignorance of experimenters in $O$ of the unknowable world beyond. 
To these familiar principles, one can add two more ideas. Namely, the description of local physics
in $O$ should be independent of spacetime or other background structures outside $O$ to the extent that (a) the description would not change if there was no
spacetime beyond $O$, and (b) it would also be unchanged if the background structures
were changed outside $O$. Here, we have in mind that $O$ is \emph{causally convex}, i.e., every causal curve with endpoints in $O$ lies entirely within $O$. 

In stating (a) and (b) we have switched from statements about the causal complement of $O$ to 
statements about its complement. Certainly no controlled experiment in $O$ can have
any direct access to the complement, as a result of causal convexity. The causal complement
of $O$ is completely unknowable (from $O$), while the causal future and past of $O$ are each
only partially knowable, because information is lost from the causal past to the causal complement,
which also supplies information to the causal future. It seems reasonable to apply the 
ignorance principle also in this case, because the description of local physics in $O$
should not require a precise knowledge of the prior history of the world.  Certainly,
this prior history may contribute to the determination of the state of physical systems
in $O$, and can be passively observed from within $O$ (e.g., astronomical observations),
but the local physical processes in $O$ should be independent of it, and should be susceptible
to experiments in which this background is controlled for or screened out, and
which can be repeated at a later stage in the history of the world provided the local
conditions are recreated. In any case, the main aim of this discussion is to motivate ideas that can be turned into precise technical statements in the framework of local covariance to yield definite consequences. In this way, we are making constructive use of our ignorance to
guide the formulation of physical theories. We now turn to the formal development of the framework.

\section{Local covariance}\label{sec:cov}

\subsection{General setting}

The discussion above provides a motivation for the formulation of
locally covariant physical theories introduced by Brunetti, Fredenhagen and Verch~\cite{BrFrVe03} and further developed in~\cite{FewVer:dynloc_theory}. The underlying ideas 
first appeared in \cite{Verch01} and \cite{Ho&Wa01}.

The fundamental idea is that a physical theory $\Af$ should be 
formulated on general spacetime backgrounds, so that to each background $\Mb$ 
there is a
mathematical object $\Af(\Mb)$ describing the theory $\Af$
in $\Mb$. From the perspective of quantum field theory
in curved spacetimes, or general relativity, it is natural to 
aim for a description on as wide a class of spacetimes as possible. 
However, it may also be motivated from ignorance of the
precise geometry or topology outside the region of immediate interest:
the description ought to be valid regardless of how the
spacetime is continued (or whether it continues at all).

The spacetimes of interest should form the objects of a category $\BG$, 
in which the morphisms indicate allowed spacetime embeddings:
$\psi:\Mb\to\Nb$ in $\BG$ indicates that $\psi$ is a way of embedding
spacetime $\Mb$ as a subspacetime of $\Nb$; equivalently,
we may think of $\psi$ as specifying a particular continuation of $\Mb$
as a spacetime, with the idea that physics in $\Mb$ should be indistinguishable
from physics within its image in $\Nb$. 

 Successive spacetime
embeddings $\varphi: \Lb\to\Mb$ and $\psi:\Mb\to\Nb$
naturally provide a means of embedding $\Lb$ in $\Nb$, which
is given by the composite morphism $\psi\circ\varphi$.
Of course, every spacetime $\Mb$ can be regarded as a subspacetime
of itself in a trivial way; this corresponds to the identity morphism
$\id_\Mb$ of $\Mb$. As we intend that a morphism from $\Mb$
to $\Nb$ embeds $\Mb$ as a subspacetime of $\Nb$ 
it is reasonable to demand that all morphisms in $\BG$ are monic;
that is, $\psi\circ\varphi_1=\psi\circ\varphi_2$ implies that $\varphi_1=\varphi_2$.

The precise specification of the category $\BG$ can vary.
The main example studied to date is the category $\Loc$,
whose objects are oriented and time-oriented globally hyperbolic 
spacetimes of fixed (but arbitrary) dimension $n\ge 2$. Each object
$\Mb$ is thus a quadruple $\Mb=(\Mc,\gb, \ogth, \tgth)$ in which
$\Mc$ is a smooth $n$-dimensional orientable manifold, equipped with
smooth Lorentz metric $\gb$ of signature $+-\cdots -$,
and orientation $\ogth$ (a component of the set of smooth nowhere vanishing $n$-form fields) and a time orientation $\tgth$ (a component of the cone of smooth nowhere vanishing timelike $1$-form fields),
so that the condition of global hyperbolicity holds: there are no
closed causal curves, and all sets of the form $J_\Mb^+(p)\cap J_\Mb^-(q)$
are compact for $p,q\in\Mc$.
 
A morphism $\psi:(\Mc,\gb,\ogth,\tgth)\to(\Mc',\gb',\ogth',\tgth')$ in $\Loc$ is given by a smooth embedding $\psi:\Mc\to\Mc'$ 
of the underlying manifolds that is isometric and preserves the
orientation and time-orientation
\[
\psi^*\gb'=\gb, \qquad \psi^*\ogth'=\ogth,\qquad \psi^*\tgth'=\tgth,
\]
and such that the image $\psi(\Mc)$ is a causally convex subset
of the codomain spacetime. In particular, $\psi$ is injective as
a function and monic as a morphism in $\BG$. 

Depending on the precise application, one could allow for additional background structure, such as background source fields (e.g., \cite{SandDappHack:2012,FewSch:2014})
or more general bundle structures (e.g., \cite{BeniniDappiaggiSchenkel:2013}).
Alternatively, one could equally allow for different models of
spacetime structure, for example, taking $\BG$ to be a category of discrete causal
sets with morphisms injective maps respecting causal order and
with a causally convex image. 

The mathematical objects describing the given theory
in specific spacetimes are also required to be objects within
a category $\Phys$. Here, the interpretation of a morphism
$f:P\to Q$ in $\Phys$ is that $f$ embeds the physical system $P$ as a subsystem
of $Q$; accordingly, we assume 
that all morphisms in $\Phys$ be monic (in some applications
this has been relaxed e.g.,~\cite{DappLang:2012,SandDappHack:2012,FewLang_Maxwell:2014,BecSchSza:2014}). The specification
of $\Phys$ reflects the type of physical theory under 
consideration. Commonly employed examples
include $\Sympl$, the category of symplectic real vector
spaces with symplectic maps as morphisms [e.g., to 
model linear dynamical systems], or the
category $\Alg$ of unital $*$-algebras with unit-preserving injective $*$-homomorphisms as morphisms, which would be a natural
setting for a description of quantum theory in terms of local
algebras of fields and/or observables. This setting is
deliberately general and allows for many variations, e.g., 
restricting to the subcategory $\CAlg$ of $\Alg$,
whose objects are required to be $C^*$-algebras. 
It should be clear that we are describing a framework
for physical theories, rather than any particular theory.

Returning to the description of a theory $\Af$,
let us consider a morphism $\psi:\Mb\to \Nb$ in $\BG$.
For each of $\Mb$ and $\Nb$, there should be a 
mathematical object $\Af(\Mb)$ and $\Af(\Nb)$
of $\Phys$. Now our basic idea is that $\psi$ embeds
$\Mb$ in $\Nb$ in such a way that physics in
$\Mb$ ought to be indistinguishable from
physics in the image $\psi(\Mb)$ in $\Nb$.
Thus, there should be a way of embedding
$\Af(\Mb)$ as a physical subsystem of $\Af(\Nb)$,
represented by a morphism $\Af(\psi):\Af(\Mb)\to
\Af(\Nb)$, which protects the ignorance of an
experimenter in $\Mb$ as to whether the
experimental region is really just a portion of
the larger spacetime $\Nb$.
There may, of course, be many ways of
embedding $\Af(\Mb)$ in $\Af(\Nb)$, but
our assumption is that the theory should
specify one of these. We require
\begin{itemize}
\item $\Af(\id_\Mb) = \id_{\Af(\Mb)}$ for every $\Mb$; i.e., trivial embeddings of
backgrounds correspond to trivial subsystem embeddings  
\item for successive embeddings $\varphi: \Lb\to\Mb$ and $\psi:\Mb\to\Nb$, the subsystem embeddings obey
\[
\Af(\psi)\circ\Af(\varphi)  =\Af(\psi\circ\varphi),
\]
i.e., the composite of the two subsystem embeddings 
should be the subsystem embedding of the composite spacetime
embedding.
\end{itemize}
The second part again reflects the ignorance principle:
if one cannot detect whether spacetime has been 
extended at all, one should certainly not be able
to distinguish whether it was extended in one step
or in two successive stages. These demands together
amount to the following definition.
\begin{definition} A locally covariant physical theory is
a covariant functor $\Af:\BG\to\Phys$.
\end{definition}

\subsection{Example: scalar field with sources}

As an example, we describe a simple theory: the
Klein--Gordon equation with an external source. This
model was studied in some detail recently in \cite{FewSch:2014}
-- the presentation here is a streamlined account 
and most details are suppressed. The discussion 
at the end of Section~\ref{sec:rce} and 
Theorem~\ref{thm:inequiv} are new.

As the category of backgrounds we take a category
$\LocSrc$ whose objects are pairs $(\Mb,\Jb)$ where
$\Mb$ is an object of $\Mb$ and $\Jb\in C^\infty(\Mb)$
is a smooth background field. A morphism from $(\Mb,\Jb)$
to $(\Mb',\Jb')$ is defined by a morphism $\psi:\Mb\to\Mb'$
in $\Loc$ that also obeys $\psi^*\Jb'=\Jb$. 

The classical theory we wish to describe has Lagrange
density 
\begin{equation}\label{eq:action}
\LL_{(\Mb,\Jb)} =\rho_\gb\left(\frac{1}{2} g^{ab}\nabla_a\phi \nabla_b\phi -
\frac{1}{2} m^2 \phi^2 -   \Jb \phi\right)
\end{equation}
where $\rho_\gb$ is the canonical volume density induced by $\gb$. The corresponding
equation of motion is
\[
(\Box_\Mb +m^2)\phi +\Jb = 0
\]
on background $(\Mb,\Jb)$. As $\Mb$ is globally hyperbolic,
there are unique advanced ($-$) and retarded ($+$) 
Green functions $E_{\Mb}^\pm:\CoinX{\Mb}\to C^\infty(\Mb)$ 
such that, for each $f\in\CoinX{\Mb}$, $\phi=E_{\Mb}^\pm f$ is the unique smooth solution
to $(\Box_\Mb+m^2)\phi=f$ with support in $J^\pm_\Mb(\supp f)$.  

For the quantum field theory, we define on each $(\Mb,\Jb)$
a unital $*$-algebra $\Af(\Mb,\Jb)$, with unit $\II_{\Af(\Mb,\Jb)}$, generated by 
elements $\Phi_{(\Mb,\Jb)}(f)$ ($f\in\CoinX{\Mb}$) and
subject to the relations
\begin{itemize}
\item Complex linearity of $f\mapsto \Phi_{(\Mb,\Jb)}(f)$
\item Hermiticity: $\Phi_{(\Mb,\Jb)}(f)^*= \Phi_{(\Mb,\Jb)}(\overline{f})$ for all $f\in\CoinX{\Mb}$
\item Field equation: 
\[
\Phi_{(\Mb,\Jb)}((\Box_\Mb + m^2)f) +
\left(\int_\Mb \Jb f \dvol_\Mb \right) \II_{\Af(\Mb,\Jb)} = 0
\]
for all $f\in\CoinX{\Mb}$
\item Commutation relation
\[{}
[\Phi_{(\Mb,\Jb)}(f),\Phi_{(\Mb,\Jb)}(f')]=i E_{\Mb}(f,f')\II_{\Af(\Mb,\Jb)}
\]
for all $f,f'\in\CoinX{\Mb}$, where
\[
E_{\Mb}(f,f')= \int_\Mb f (E_\Mb^-- E_\Mb^+)f' \dvol_\Mb.
\]
\end{itemize}
The interpretation of this algebra is that the generator
$\Phi_{(\Mb,\Jb)}(f)$ is to be thought of as
the quantum field smeared against test function $f$.
The form of the commutation relations may be 
motivated as the Dirac quantization of the 
classical Peierls' bracket of the corresponding
smeared classical fields. 

This completes the construction of the algebra in each spacetime.
Given a morphism $\psi:(\Mb,\Jb)\to(\Mb',\Jb')$ in $\LocSrc$,
we may define a map $\Af(\psi):\Af(\Mb,\Jb)\to\Af(\Mb',\Jb')$ by
\[
\Af(\psi)\Phi_{(\Mb,\Jb)}(f) = \Phi_{(\Mb',\Jb')}(\psi_*f)
\]
for $f\in\CoinX{\Mb}$, where $\psi_*$ denotes the push-forward
of compactly supported functions
\[
(\psi_*f)(p)=\begin{cases} f(\psi^{-1}(p)) & p\in\psi(\Mb)\\
0 & \text{otherwise.}\end{cases}
\]
The above expression defines $\Af(\psi)$ on the generators
of $\Af(\Mb,\Jb)$ and, because it is compatible with
the relations imposed in the two algebras, it extends to a unit-preserving
$*$-homomorphism, which is injective ($\Af(\Mb,\Jb)$ can be shown to be a simple
algebra). Thus $\Af(\psi)$ is a morphism in $\Alg$. It is clear from the definition
and properties of the push-forward that 
that the functorial conditions $\Af(\id_{(\Mb,\Jb)})=\id_{\Af(\Mb,\Jb)}$ 
and $\Af(\psi\circ\varphi)=\Af(\psi)\circ\Af(\varphi)$ are met. Accordingly,
we have defined the theory as a functor $\Af:\LocSrc\to\Alg$.

The definition of the morphisms $\Af(\psi)$ seems almost an afterthought,
but it is actually crucial to the definition of the theory. Indeed, the
algebras by themselves do not really specify the theory at all. To see this,
fix a background $(\Mb,\Jb)$ and also choose a particular real-valued solution $\phi\in C^\infty(\Mb)$
to the classical equation of motion. Now define a map $\kappa_\phi:\Af(\Mb,\Jb)\to\Af(\Mb,0)$
by 
\begin{equation}\label{eq:kappaphi}
\kappa_\phi \Phi_{(\Mb,\Jb)}(f) = \Phi_{(\Mb,0)}(f) + \left(\int_\Mb f\phi\,\dvol_\Mb\right)\II_{\Af(\Mb,0)}
\end{equation}
and $\kappa_\phi  \II_{\Af(\Mb,\Jb)}=\II_{\Af(\Mb,0)}$. One may check that this
map on generators is compatible with the relations of the two algebras\footnote{
For example, let $A=\Phi_{(\Mb,\Jb)}((\Box_\Mb+m^2)f) + \left(\int_\Mb \Jb f \dvol_\Mb \right) \II_{\Af(\Mb,\Jb)}$. Then our definitions give
\[
\kappa_\phi A= \Phi_{(\Mb,0)}((\Box_\Mb+m^2)f) + \left(\int_\Mb \left((\Box_\Mb+m^2)f\right)\phi\,\dvol_\Mb\right)\II_{\Af(\Mb,0)}  + \left(\int_\Mb \Jb f\,\dvol_\Mb\right)\II_{\Af(\Mb,0)} =0
\]
using the relations in $\Af(\Mb,0)$ and the field equation obeyed by $\phi$, 
which is consistent with the fact that $A=0$ by the relations in $\Af(\Mb,\Jb)$.} 
and therefore
extends to a morphism in $\Alg$; in fact, it is an isomorphism, with inverse defined
on generators by
\[
\kappa_\phi^{-1} \Phi_{(\Mb,0)}(f) =  \Phi_{(\Mb,\Jb)}(f) - \left(\int_\Mb f\phi\,\dvol_\Mb\right)\II_{\Af(\Mb,\Jb)}.
\]
As the algebras $\Af(\Mb,\Jb)$ and $\Af(\Mb,0)$ are isomorphic, they 
carry no specific information about the background source $\Jb$.\footnote{By the time-slice property (discussed later) there are isomorphisms $\Af(\Mb,\Jb)\cong \Af(\Mb',\Jb')$ whenever
the Cauchy surfaces of $\Mb$ and $\Mb'$ can be related by an orientation-preserving
diffeomorphism.} Note, however, that $\kappa_\phi$ depends on the special
choice of a particular solution $\phi$, and there is no canonical way of choosing
such a solution in a general background. 

Developing this point a bit further, let $\lambda\in\RR$ and define a functor
$\Zf_\lambda:\LocSrc\to\LocSrc$ so that
$\Zf_\lambda(\Mb,\Jb)=(\Mb,\lambda\Jb)$ and so that, if $\psi:(\Mb,\Jb)\to (\Mb',\Jb')$, then
 $\Zf_\lambda(\psi):(\Mb,\lambda\Jb)\to (\Mb',\lambda\Jb')$ has the same underlying map as $\psi$. Then the functor $\Af_\lambda:=\Af\circ\Zf_\lambda : \LocSrc \to \Alg$ is a new theory,
which assigns algebra $\Af_\lambda(\Mb,\Jb)=\Af(\Mb,\lambda \Jb)$ to the background $(\Mb,\Jb)$;
namely, $\Af_\lambda$ is the theory of the inhomogeneous field with a coupling strength $\lambda$, 
corresponding to classical field equation $(\Box_\Mb+m^2)\phi+\lambda\Jb=0$ on background $(\Mb,\Jb)$. For $\lambda\neq \mu$, we would expect that the theories $\Af_\lambda$ and
$\Af_{\mu}$ should represent different physics; however, the arguments above show that
$\Af_\lambda(\Mb,\Jb)$ and $\Af_\mu(\Mb,\Jb)$ are isomorphic for each background.  
Therefore it is the specification of the morphisms of the theories, rather than the objects (algebras,
in this case) that distinguishes them. We return to this point in   subsection~\ref{sec:rce} below.

The theory $\Af$ determines, on each background, the algebra of smeared fields of the inhomogeneous scalar field, suitable elements of which are observables. To complete the physical description we need to specify allowed states -- our discussion is based on \cite{BrFrVe03, Fewster:gauge}. A {\em state space} for $\Ac\in\Alg$ is a subset $\Sc$ of linear functionals $\omega$ on $\Ac$ that are 
normalized ($\omega(\II_\Ac)=1$), positive ($\omega(A^*A)\ge 0$ for all $A\in\Ac$) 
and so that $\Sc$ is closed under convex linear combinations.\footnote{For simplicity
of presentation, we suppress a further condition often imposed on state spaces: namely
that $\Sc$ should be closed under operations induced by $\Ac$, i.e.,
to each $\omega\in\Sc$ and $B\in\Ac$ with $\omega(B^*B)> 0$, the state $\omega_B(A):=\omega(B^*AB)/\omega(B^*B)$ is also an element of $\Sc$.}
As usual, $\omega(A)$ is interpreted as the expectation value of observable $A$
in state $\omega$ --- using the GNS theorem, each state induces
a Hilbert space representation $\pi_\omega$ of $\Ac$ on a Hilbert space $\HH_\omega$
with a distinguished vector $\Omega_\omega\in\HH_\omega$ so that $\omega(A)=\ip{\Omega_\omega}{\pi(A)\Omega_\omega}$ for all $A\in\Ac$, recovering the familiar Born probability rule.
In particular, the set of all states $\Ac^*_{+,1}$ of $\Ac$ is a state space. 

We may now introduce a category $\AlgSts$, whose objects are pairs $(\Ac,\Sc)$,
where $\Ac\in\Alg$ and $\Sc$ is a state space for $\Ac$. A morphism in $\AlgSts$
between objects $(\Ac,\Sc)$ and $(\Bc,\Tc)$ is an $\Alg$-morphism $\alpha:\Ac\to\Bc$
with the additional property $\alpha^*\Tc \subset \Sc$, where $\alpha^*$ is the dual map
to $\alpha$. Composition of morphisms in $\AlgSts$ is inherited from $\Alg$. 

Our inhomogeneous scalar field theory $\Af$ may be augmented to a theory with
values in $\AlgSts$ in various ways. The simplest is to define $\tilde{\Af}:\LocSrc\to\AlgSts$
so that  $\tilde{\Af}(\Mb,\Jb)= (\Af(\Mb,\Jb),\Af(\Mb,\Jb)^*_{+,1})$ for each object $(\Mb,\Jb)$
of $\LocSrc$,
and taking $\tilde{\Af}(\psi)$ to be the morphism induced by $\Af(\psi)$ for each morphism
$\psi:(\Mb,\Jb)\to(\Mb',\Jb')$ in $\LocSrc$  (note that $\Af(\psi)^*$ maps any state
of $\Af(\Mb',\Jb')$ to a state of $\Af(\Mb,\Jb)$). A more interesting possibility is to equip each $\Ac(\Mb,\Jb)$ with the corresponding set of Hadamard states, which are the standard choice
of physically acceptable states of the scalar field (see, e.g.,~\cite{Wald_qft}).

\begin{definition} A state $\omega$ on $\Ac(\Mb,\Jb)$ is said to be \emph{Hadamard}
if the corresponding two-point function $W^{(2)}_\omega:\CoinX{\Mb}\times\CoinX{\Mb}\to\CC$
defined by $W^{(2)}_\omega(f,f')= \omega(\Phi_{(\Mb,\Jb)}(f)\Phi_{(\Mb,\Jb)}(f'))$ 
is a distribution in $\DD'(\Mb\times\Mb)$ with wave-front set
\begin{equation}\label{eq:Hadamard}
\WF(W^{(2)}_\omega) \subset \Nc^-\times \Nc^+ 
\end{equation}
where $\Nc^{+/-}$ is the bundle of future/past-directed null covectors on $\Mb$. 
The set of all Hadamard states on $\Ac(\Mb,\Jb)$ will
be denoted $\Sc(\Mb,\Jb)$. 
\end{definition}
It would take us too far from our main purpose to give the definition
of the wave-front set here (see \cite{Hormander1} for details)
but the main points are that:
\begin{itemize}
\item the wave-front set $\WF(u)$ of a distribution $u\in\DD'(X)$ on manifold $X$
is a subset of the cotangent bundle $T^*X$ encoding the singular structure of $u$ ---
in particular smooth distributions have empty wave-front sets;
\item under pull-backs by smooth functions the wave-front set obeys
\[
\WF(\kappa^*u)\subset \kappa^*\WF(u);
\]
\item in the QFT context, the wave-front set condition \eqref{eq:Hadamard} 
on the two-point function is sufficient to fix the wave-front sets of all $n$-point
functions exactly (combining \cite[Prop.~6.1]{StrVerWol:2002} and \cite{Sanders_uSC:2010}) and also
ensures that the two-point function differs from the `Hadamard parametrix' by a smooth function
\cite{Radzikowski_ulocal1996}; 
\item the form of the wave-front set condition recalls fact that Minkowski two-point functions
are positive frequency in the first variable and negative frequency in the second.\footnote{The 
correspondence between `positive frequency' and $\Nc^-$ is an unfortunate by-product
of the standard conventions for Fourier transform, used in \cite{Hormander1}. 
Ref.~\cite{Radzikowski_ulocal1996} and some of the other literature use nonstandard
conventions to remove this issue.}
\end{itemize}
It is easily seen that $\Sc(\Mb,\Jb)$ is a state space for $\Ac(\Mb,\Jb)$. 
Now consider a morphism $\psi:(\Mb,\Jb)\to(\Mb',\Jb')$ and a Hadamard state
$\omega\in\Sc(\Mb',\Jb')$. Then the $n$-point functions of $\omega$ and
$\Af(\psi)^*\omega$ are related by
\begin{align*}
W_{\Af(\psi)^*\omega}^{(n)}(f_1,\ldots,f_n) &= (\Af(\psi)^*\omega)(\Phi_{(\Mb,\Jb)}(f_1)
\cdots \Phi_{(\Mb,\Jb)}(f_n)) \\
&= \omega((\Af(\psi)\Phi_{(\Mb,\Jb)}( f_1)) \cdots (\Af(\psi)\Phi_{(\Mb,\Jb)}( f_n)))\\
&= \omega(\Phi_{(\Mb',\Jb')}(\psi_* f_1) \cdots \Phi_{(\Mb',\Jb')}(\psi_* f_n)) = W_{\omega}^{(n)}(\psi_* f_1,\ldots,\psi_* f_n) 
\end{align*}
i.e., $W_{\Af(\psi)^*\omega}^{(n)}=(\psi\times\cdots\times\psi)^*W_\omega^{(n)}$, from which
it follows that $\Af(\psi)^*\omega$ is Hadamard. Hence, setting $\tilde{\Af}(\Mb,\Jb) = 
(\Af(\Mb,\Jb),\Sc(\Mb,\Jb))$, $\Af(\psi)$ induces a morphism $\tilde{\Af}(\psi)$ in $\AlgSts$ between $\tilde{\Af}(\Mb,\Jb)$ and $\tilde{\Af}(\Mb',\Jb')$. This defines a new theory
$\tilde{\Af}:\LocSrc\to\AlgSts$. 

One might reasonably wonder how small a state space can be: can we choose
a state space consisting of a single state $\omega_{(\Mb,\Jb)}$ for each background,
so that $\Af(\psi)^*\omega_{(\Mb',\Jb')}=\omega_{(\Mb,\Jb)}$ for every morphism
$\psi:(\Mb,\Jb)\to(\Mb',\Jb')$? Quantum field theory in Minkowski space is
so tightly framed around the vacuum state that it is only natural to seek
a replacement in general spacetime backgrounds. However, 
as will be discussed later, such \emph{natural states}
can be ruled out under general circumstances for quantum field theories.   
This has two consequences for physics. First, the particle interpretation
of the theory is based on excitations of `the vacuum', and so
the loss of a preferred state also means the loss of a preferred
notion of particles. Second, any procedure that does define a state in all
spacetime, such as a path integral prescription [setting aside the difficulties of making this precise] or the recently proposed construction~\cite{AAS}, must depend in a nonlocal way on the
spacetime including those portions outside the experimental region
controlled by an observer. It would seem to run counter to
the ignorance meta-principle to ascribe operational significance
to such a state. Further discussion on these lines appears in~\cite[\S 5]{CF-RV-ultraSJ}.

\subsection{Relations between theories}

The example of the theories $\Af_\lambda$ shows that physical equivalence of
two theories $\Af,\Bf:\BG\to\Phys$ is not simply a matter of the existence of isomorphisms
between $\Af(\Mb)$ and $\Bf(\Mb)$ for each background $\Mb$.
\begin{definition}
A theory $\Af:\BG\to\Phys$ is a \emph{subtheory} of $\Bf:\BG\to\Phys$ if there is
a natural transformation $\zeta:\Af\nto\Bf$. The theories are \emph{equivalent} if 
there is a natural isomorphism between them.
\end{definition}
Recall that a natural transformation between functors $\Af$ and $\Bf$ is a 
collection of $\Phys$-morphisms $\zeta_\Mb:\Af(\Mb)\to\Bf(\Mb)$ for each $\Mb\in\BG$ 
such that the diagram
\[
\begin{tikzpicture}[baseline=0 em, description/.style={fill=white,inner sep=2pt}]
\matrix (m) [ampersand replacement=\&,matrix of math nodes, row sep=3em,
column sep=2.5em, text height=1.5ex, text depth=0.25ex]
{ \Af(\Mb) \&  \Bf(\Mb) \\
\Af(\Nb) \&  \Bf(\Nb)\\ };
\path[->,font=\scriptsize]
(m-1-1) edge node[auto] {$ \zeta_\Mb $} (m-1-2)
        edge node[auto] {$ \Af(\psi) $} (m-2-1)
(m-2-1) edge node[auto] {$ \zeta_\Nb $} (m-2-2)
(m-1-2) edge node[auto] {$ \Bf(\psi) $} (m-2-2);
\end{tikzpicture} 
\]
commutes whenever $\psi:\Mb\to\Nb$ is a morphism in $\BG$.
In other words, the operations of passing between spacetimes and passing between theories
must commute. For $\zeta$ to be a natural isomorphism, 
each $\zeta_\Mb$ must be a $\Phys$-isomorphism.
In the case where $\Af$ and $\Bf$ coincide, the 
natural automorphisms of $\Af$ turn out to have a
physically natural interpretation: they are the global gauge 
transformations of the theory~\cite{Fewster:gauge}.  

In Theorem~\ref{thm:inequiv} we will show that
the theories  $\Af_\lambda$ are inequivalent for distinct $\lambda\in\RR$. 
As a mild digression, we show that this rules out the existence of natural
states in the theory. For suppose that 
the theory $\Af_\lambda$ admits a natural state $(\omega_{(\Mb,\Jb)})_{(\Mb,\Jb)\in\LocSrc}$ for some
$\lambda\neq 0$.  In each background $(\Mb,\Jb)\in\LocSrc$,  the 
one-point function of the natural state $W^{(1)}_{(\Mb,\Jb)}(f) = \omega_{(\Mb,\Jb)}(\Phi_{(\Mb,\Jb)}(f))$  solves the
classical field equation in the sense that
\[
W^{(1)}_{(\Mb,\Jb)}((\Box_\Mb+m^2)f) + \int_\Mb \Jb f\,\dvol_\Mb = 0.
\]
By analogy with \eqref{eq:kappaphi} we may define
an $\Alg$-isomorphism $\kappa_{(\Mb,\Jb)}:\Af_\lambda(\Mb,\Jb)\to
\Af_0(\Mb,\Jb)$ acting on generators by
\[
\kappa_{(\Mb,\Jb)} \Phi_{(\Mb,\lambda\Jb)}(f) = \Phi_{(\Mb,0)}(f) + W^{(1)}_{(\Mb,\Jb)}(f)\II_{\Af(\Mb,0)}
\]
(recall that $\Af_\lambda(\Mb,\Jb)=\Af(\Mb,\lambda\Jb)$
has generators $\Phi_{(\Mb,\lambda\Jb)}(f)$, $f\in\CoinX{\Mb}$).
The morphisms $\kappa_{(\Mb,\Jb)}$ yield a natural isomorphism
between $\Af_\lambda$ and $\Af_0$. To see this, consider a morphism
$\psi:(\Mb,\Jb)\to(\Mb',\Jb')$.  For any $f\in\CoinX{\Mb}$,
\[
\Af_0(\psi)\circ \kappa_{(\Mb,\Jb)}  \Phi_{(\Mb,\lambda\Jb)}(f) 
=  \Phi_{(\Mb',0)}(\psi_*f) + W^{(1)}_{(\Mb,\Jb)}(f)\II_{\Af(\Mb',0)},
\]
while 
\[
\kappa_{(\Mb',\Jb')}\circ \Af_\lambda(\psi) \Phi_{(\Mb,\lambda\Jb)}(f) 
=  \Phi_{(\Mb',0)}(\psi_*f) + W^{(1)}_{(\Mb',\Jb')}(\psi_*f)\II_{\Af(\Mb',0)}.
\]
However, the naturality of the state entails precisely that
\begin{align*}
W^{(1)}_{(\Mb',\Jb')}(\psi_*f)&=
\omega_{(\Mb',\Jb')}(\Phi_{(\Mb',\lambda\Jb')}(\psi_*f)) = 
\omega_{(\Mb',\Jb')}(\Af_\lambda(\psi)\Phi_{(\Mb,\lambda\Jb)}(f))\\
&=\omega_{(\Mb,\Jb)}(\Phi_{(\Mb,\lambda\Jb)}(f))
=
W^{(1)}_{(\Mb,\Jb)}(f).
\end{align*}
As $f$ was arbitrary, it follows that 
\[
\Af_0(\psi)\circ \kappa_{(\Mb,\Jb)}= 
\kappa_{(\Mb',\Jb')}\circ \Af_\lambda(\psi)
\]
which shows that the $\kappa_{(\Mb,\Jb)}$ cohere to
form a natural isomorphism $\kappa:\Af_\lambda\nto \Af_0$. 
If natural states existed for all the theories $\Af_\lambda$ 
($\lambda\in\RR$) then one would be able to establish equivalence between 
all of them. Thus, the inequivalence of these theories 
precludes the existence of natural states.

\subsection{Time-slice axiom and relative Cauchy evolution}
\label{sec:rce}

So far, we have only imposed the condition of local covariance on physical
theories. A much stronger condition is the time-slice axiom, which can
be regarded as encoding the existence of a dynamical law in the theory.
For our presentation, we restrict to categories such as $\Loc$ and
$\LocSrc$ that are based on globally hyperbolic spacetimes, but
the ideas can be extended to more general settings. A morphism $\psi:\Mb\to\Mb'$
in $\Loc$ whose image $\psi(\Mb)$ contains a Cauchy surface of $\Mb'$ will be 
called a \emph{Cauchy morphism}; similarly, a morphism $\psi:(\Mb,\Jb)\to
(\Mb',\Jb')$ in $\LocSrc$ is called Cauchy under the same condition.
\begin{definition}
A locally covariant theory $\Af:\BG\to\Phys$ (where $\BG$ is $\Loc$ or $\LocSrc$) obeys the 
\emph{time-slice axiom}
if $\Af$ maps every Cauchy morphism of $\BG$ to an isomorphism in $\Phys$.
\end{definition}
If $\psi:\Mb\to\Mb'$ is a Cauchy morphism, every aspect of the
physics on $\Mb'$ can be predicted from the physics on $\Mb$,
provided the time-slice axiom holds. 

It was realized by BFV that the time-slice axiom allows the
comparison of dynamics on different backgrounds, in terms of
\emph{relative Cauchy evolution}~\cite{BrFrVe03}. Here we 
describe the adaptation to $\LocSrc$ given in~\cite{FewSch:2014},
with some slight modifications. 
Let $(\Mb,\Jb)$ be an object of $\LocSrc$, with $\Mb=(M,\gb, \ogth, \tgth)$. Let $\hb$ be a compactly supported rank-$2$ covariant
symmetric tensor field such that $(M,\gb+\hb)$ is a globally 
hyperbolic spacetime, with respect to the (unique) time-orientation
$\tgth[\hb]$ that agrees with $\tgth$ outside $\supp\hb$. Then
$\Mb[\hb]=(M,\gb+\hb,\ogth,\tgth[\hb])$ is an object of $\Loc$, 
and $(\Mb[\hb], \Jb[\jb]):=(\Mb[\hb], \Jb+\jb)$ is an object of $\LocSrc$
for any $\jb\in\CoinX{\Mb}$.  
Under these circumstances, we write $(\hb,\jb)\in H(\Mb,\Jb)$. 
Choose open $\gb$-causally convex sets $M^{+/-}$ of $M$ 
lying to the future/past of the $\supp(\hb)\cup \supp(\jb)$,\footnote{That is, there should be Cauchy surfaces $S^\pm$ such that 
$M^\pm\subset I^\pm_\Mb(S^\pm)$, 
$\supp(\hb)\cup \supp(\jb)\subset I^\mp_\Mb(S^\pm)$.} and containing Cauchy surfaces of $\Mb$, as in Fig.~\ref{fig:rce}. 
These sets are therefore also causally convex with respect to
$\gb+\hb$ and, defining $\Mb^\pm$
to be the sets $M^\pm$ equipped with causal structure and orientation 
inherited from $\Mb$, and $\Jb^\pm=\Jb|_{M^\pm}$, 
there are $\LocSrc$ Cauchy morphisms 
\begin{subequations}
\begin{flalign}
i^\pm &: \big(\Mb^\pm,\Jb^\pm\big) \to \big(\Mb,\Jb\big)~,\\
j^\pm &: \big(\Mb^\pm,\Jb^\pm\big) \to \big(\Mb[\hb],\Jb[\jb]\big)
\end{flalign}
\end{subequations}
induced by the set inclusions of $M^\pm$ in $M$. 
A theory $\Af:\LocSrc\to\Phys$ that obeys the time-slice axiom 
converts each of these morphisms to an isomorphism. 
The relative Cauchy evolution of $\Af$ induced by $(\hb,\jb)\in H(\Mb,\Jb)$ is defined as the automorphism
\begin{equation}\label{eq:rce_def}
\rce_{(\Mb,\Jb)}[\hb,\jb]:=
\Af(i^+ )\circ \Af(j^- )^{-1}\circ
\Af(j^+ )\circ \Af(i^+ )^{-1}    
\end{equation}
of $\Af(\Mb,\Jb)$. One may show that $\rce_{(\Mb,\Jb)}[\hb,\jb]$ is independent of the choices of $M^\pm$ made in the construction, cf. \cite[\S 3]{FewVer:dynloc_theory}.

The significance of relative Cauchy evolution can be explained
as follows. Owing to the dynamical law of the theory on $(\Mb,\Jb)$, 
any observable $A$ can be measured in the region $M^+$. 
We fix that description, and proceed to modify the background
to the past of $M^+$, and future of $M^-$, obtaining an observable
in the $M^+$ region of the perturbed background $(\Mb[\hb],\Jb[\jb])$. 
In turn, this observable can also be measured in the $M^-$ region of 
$(\Mb[\hb],\Jb[\jb])$, owing to the dynamical law of the theory on the perturbed background. Transferring that description to 
the $M^-$ region of the unperturbed spacetime, we obtain a new
observable on $(\Mb,\Jb)$ which will not in general coincide with
our original observable $A$. The discrepancy is precisely measured
by the relative Cauchy evolution, which applies in this way to all
aspects of the theory, not just observables. 

\begin{figure}
\begin{center}
\begin{tikzpicture}[scale=0.9]
\definecolor{Gold}{rgb}{.93,.82,.24}
\definecolor{Orange}{rgb}{1,0.5,0}
\draw[fill=lightgray] (-4,0) -- ++(2,0) -- ++(0,4) -- ++(-2,0) -- cycle;
\draw[fill=lightgray] (4,0) -- ++(2,0) -- ++(0,4) -- ++(-2,0) -- cycle;
\draw[fill=gray] (4,3) -- ++(2,0) -- ++(0,0.5) -- ++(-2,0) -- cycle;
\draw[fill=gray] (0,3) -- ++(2,0) -- ++(0,0.5) -- ++(-2,0) -- cycle;
\draw[fill=gray] (-4,3) -- ++(2,0) -- ++(0,0.5) -- ++(-2,0) -- cycle;
\draw[fill=gray] (4,0.5) -- ++(2,0) -- ++(0,0.5) -- ++(-2,0) -- cycle;
\draw[fill=gray] (0,0.5) -- ++(2,0) -- ++(0,0.5) -- ++(-2,0) -- cycle;
\draw[fill=gray] (-4,0.5) -- ++(2,0) -- ++(0,0.5) -- ++(-2,0) -- cycle;
\draw[color=black,line width=2pt,->] (2.25,3.25) -- (3.75,3.25) node[pos=0.4,above]{$j^+$};
\draw[color=black,line width=2pt,->] (2.25,0.75) -- (3.75,0.75) node[pos=0.4,above]{$j^-$};
\draw[color=black,line width=2pt,->] (-0.25,3.25) -- (-1.75,3.25) node[pos=0.5,above]{$i^+$};
\draw[color=black,line width=2pt,->] (-0.25,0.75) -- (-1.75,0.75) node[pos=0.5,above]{$i^-$};
\draw[fill=white] (5,2) ellipse (0.7 and 0.4);
\node at (5,2) {$\hb,\jb$};
\node[anchor=north] at (5,0) {$(\Mb[\hb],\Jb[\jb])$};
\node[anchor=north] at (-3,0) {$(\Mb,\Jb)$};
\node[anchor=north] at (1,3) {$(\Mb^+,\Jb^+)$};
\node[anchor=north] at (1,0.5) {$(\Mb^-,\Jb^-)$};
\end{tikzpicture}
\end{center}
\caption{Schematic representation of relative Cauchy evolution}\label{fig:rce}
\end{figure}
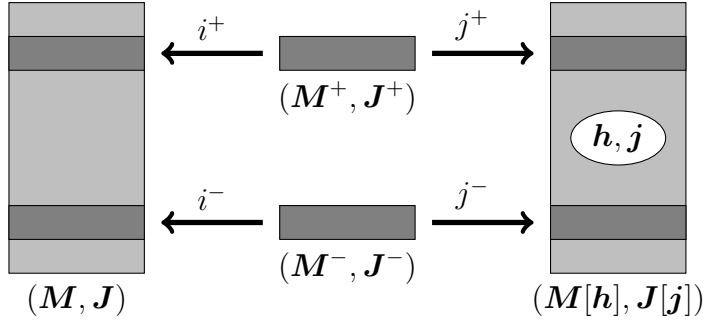

The relative Cauchy evolution is particularly interesting for infinitesimal perturbations:
its functional derivative with respect to the background metric yields
a derivation related to the stress-energy tensor~\cite{BrFrVe03} and similar
results are obtained for other background sources. For example, in  the case of the 
the inhomogeneous scalar field theory one finds that 
\cite[\S 7.3]{FewSch:2014} 
\begin{equation}\label{eq:diffrce}
\left.\frac{d}{ds} \rce_{(\Mb,\Jb)}[s\hb,s\jb] A\right|_{s=0}  = i\,\left[\frac{1}{2}  T_{(\Mb,\Jb)}(\hb)+ \Phi_{(\Mb,\Jb)}(\jb), A\right],
\end{equation}
for all $A\in\Af(\Mb,\Jb)$, where $T_{(\Mb,\Jb)}(\hb) = \int_M h_{ab}\,T^{ab}_{(\Mb,\Jb)}\,\dvol_{\Mb}$ is the smearing with $h_{ab}$ 
of the quantization of the stress-energy tensor\footnote{
The renormalized stress-energy tensor is not an element of
$\Af(\Mb,\Jb)$, but we write $A\mapsto \left[T_{(\Mb,\Jb)}(\hb),A\right]$
as convenient notation for the outer derivation of $\Af(\Mb,\Jb)$
obtained by regularizing the stress-energy tensor by point-splitting,
computing the commutator within $\Af(\Mb,\Jb)$ and then
removing the regulation. In \cite[\S 7.3]{FewSch:2014}, 
the analogue of \eqref{eq:diffrce} was stated only for the 
case $A=\Phi_{(\Mb,\Jb)}(f)$, but it extends immediately to the form given here.}  
\begin{flalign}\label{eqn:fullSET}
T_{(\Mb,\Jb)}^{ab}[\phi] :=  -\frac{2}{\sqrt{\vert \gb\vert }} \frac{\delta S}{\delta g_{ab}(x)}=  \nabla^a \phi \nabla^b \phi - \frac{1}{2}g^{ab}\nabla_c \phi \nabla^c \phi
+\frac{1}{2}m^2 g^{ab} \phi^2  +
g^{ab}  \Jb \phi ,
\end{flalign}
and $S$ is the classical action obtained from the Lagrangian 
\eqref{eq:action}. In \eqref{eq:diffrce} the derivative is 
understood with respect to be taken as a weak derivative
in suitable Hilbert space representations~\cite{BrFrVe03}.  
Thus, the relative Cauchy evolution can be taken as a replacement
for the classical action, because its functional derivative corresponds
to quantities normally obtained from the functional derivative of the action
with respect to the background. 

A key observation \cite[Prop.~3.8]{FewVer:dynloc_theory} is the following.
\begin{theorem}\label{thm:rce_intertwine}
If $\Af:\LocSrc\to\Phys$ and $\Bf$ are two locally covariant theories,
obeying the time-slice axiom, and $\zeta:\Af\nto\Bf$ embeds $\Af$
as a subtheory of $\Bf$, then
\[
\zeta_{(\Mb,\Jb)}\circ \rce^{(\Af)}_{(\Mb,\Jb)}[\hb,\jb] = 
\rce^{(\Bf)}_{(\Mb,\Jb)}[\hb,\jb]
\circ\zeta_{(\Mb,\Jb)}
\]
for all background perturbations $(\hb,\jb)\in H(\Mb,\Jb)$. 
\end{theorem}
This provides a strong and practical constraint that can be 
used to rule out or classify natural transformations between theories
and particularly their automorphisms~\cite{Fewster:gauge}.  
Here, we indicate how it distinguishes the inhomogeneous 
theories $\Af_\lambda$ for different values of the coupling constant $\lambda$.
In addition to the differential formula~\eqref{eq:diffrce} it
will be useful to use the formula 
\begin{equation}\label{eq:rce}
\rce^{(\Af)}_{(\Mb,\Jb)}[\Ob,\jb]\Phi_{(\Mb,\Jb)}(f)
= \Phi_{(\Mb,\Jb)}(f) + E_\Mb(f,\jb)\II_{\Af(\Mb,\Jb)}
\end{equation}
also obtained in~\cite{FewSch:2014}.
Now the theory $\Af_\lambda$ is defined so that $\Af_\lambda(\Mb,\Jb) = 
\Af(\Mb,\lambda\Jb)$. Hence 
\[
\rce^{(\Af_\lambda)}_{(\Mb,\Jb)}[\hb,\jb]= \rce^{(\Af)}_{(\Mb,\lambda\Jb)}[\hb,\lambda\jb],
\] 
the analogue of \eqref{eq:rce} is
\begin{equation}\label{eq:rce2}
\rce^{(\Af_\lambda)}_{(\Mb,\Jb)}[\Ob,\jb]\Phi_{(\Mb,\lambda\Jb)}(f)
= \Phi_{(\Mb,\lambda\Jb)}(f) + \lambda E_\Mb(f,\jb)\II_{\Af(\Mb,\lambda\Jb)}
\end{equation}
and \eqref{eq:diffrce} entails the formula
\begin{equation}\label{eq:diffrce2}
\left.\frac{d}{ds} \rce_{(\Mb,\Jb)}^{(\Af_\lambda)}[\Ob,s\jb]A\right|_{s=0}  = i\lambda\left[ \Phi_{(\Mb,\lambda\Jb)}(\jb), A\right].
\end{equation}

\begin{theorem} \label{thm:inequiv}
If $\lambda,\mu\in\RR$ are distinct, then
$\Af_\lambda$ and $\Af_\mu$ are inequivalent. 
\end{theorem}
\begin{proof}
Without loss, assume $\mu\neq 0$ and that there is a natural transformation $\zeta:\Af_\lambda\to\Af_\mu$. For
each background $(\Mb,\Jb)$ and all $\jb,f\in\CoinX{\Mb}$, Theorem~\ref{thm:rce_intertwine} gives 
\[
\zeta_{(\Mb,\Jb)}\circ \rce^{(\Af_\lambda)}_{(\Mb,\Jb)}[\Ob,\jb]\Phi_{(\Mb,\lambda\Jb)}(f) = 
\rce^{(\Af_\mu)}_{(\Mb,\Jb)}[\Ob,\jb]
\circ\zeta_{(\Mb,\Jb)}\Phi_{(\Mb,\lambda\Jb)}(f) 
\]
which, together with~\eqref{eq:rce2}, yields
\[
\zeta_{(\Mb,\Jb)} \Phi_{(\Mb,\lambda\Jb)}(f)+ \lambda E_\Mb(f,\jb)\II_{\Af(\Mb,\mu\Jb)}  = 
\rce^{(\Af_\mu)}_{(\Mb,\Jb)}[\Ob,\jb]
\circ\zeta_{(\Mb,\Jb)}\Phi_{(\Mb,\lambda\Jb)}(f) .
\]
However, $\mu E_\Mb(f,\jb)\II_{\Af(\Mb,\mu\Jb)} = 
\rce^{(\Af_\mu)}_{(\Mb,\Jb)}[\Ob,\jb]\Phi_{(\Mb,\mu\Jb)}(f)
- \Phi_{(\Mb,\mu\Jb)}(f)$, so after rearrangement, we may deduce that
$X(f)=\zeta_{(\Mb,\Jb)} \Phi_{(\Mb,\lambda\Jb)}(f)-(\lambda/\mu) \Phi_{(\Mb,\mu\Jb)}(f)$ 
obeys
\[
\rce^{(\Af_\mu)}_{(\Mb,\Jb)}[\Ob,\jb] X(f) = X(f)
\]
for all $\jb\in\CoinX{\Mb}$. By \eqref{eq:diffrce2}, it follows that $[\Phi_{(\Mb,\mu\Jb)}(\jb),X(f)]=0$
for all $\jb\in\CoinX{\Mb}$ and that $X(f)$ is central. Hence
\begin{align*}
[\zeta_{(\Mb,\Jb)} \Phi_{(\Mb,\lambda\Jb)}(f), \zeta_{(\Mb,\Jb)} \Phi_{(\Mb,\lambda\Jb)}(f')]
&=\left(\frac{\lambda}{\mu}\right)^2
[ \Phi_{(\Mb,\mu\Jb)}(f),  \Phi_{(\Mb,\mu\Jb)}(f')] \\
&=\left(\frac{\lambda}{\mu}\right)^2 \zeta_{(\Mb,\Jb)} 
[\Phi_{(\Mb,\lambda\Jb)}(f), \Phi_{(\Mb,\lambda\Jb)}(f')],
\end{align*}
contradicting the assumption that $\zeta_{(\Mb,\Jb)}$ is a homomorphism. Hence $\Af_\lambda$ cannot be embedded as a subtheory of $\Af_\mu$, and in particular is not equivalent to it.
\end{proof}
This result demonstrates the power of the functorial viewpoint; 
recall that the individual algebras $\Af_\lambda(\Mb,\Jb)$ and $\Af_\mu(\Mb,\Jb)$ are isomorphic for all $\lambda,\mu\in\RR$. 
Thus, the distinction between the theories lies
in the fact that there is no way of choosing such isomorphisms in 
a natural way. In this sense, the physics represented by a theory is encoded in its functorial structure. 

\section{Local physical content and dynamical locality}\label{sec:dynloc}

We turn to the description of the local physical content of a theory $\Af:\Loc\to\Phys$, in a causally convex region $O$ of spacetime $\Mb$
(there seems to be no obstruction to generalizing the background category if desired). Two ignorance principles
apply: ignorance of whether the spacetime extends at all beyond $O$, and ignorance of what
the metric might be outside $O$. According to the ignorance meta-principle, the local
content associated with $O$ ought to be independent of each factor, and suggests
two characterizations. 

First, we may consider the local physical content to be the full content we would have
if the spacetime were coterminous with $O$. To quantify this idea, 
let $\iota_O$ be the embedding $O\hookrightarrow \Mb$, which induces a $\Loc$-morphism
$\iota_O:\Mb|_O\to\Mb$, where $\Mb|_O$ is the set $O$ with the metric and causal
structure induced from $\Mb$.  
Then the functor describing the theory assigns a morphism
$\Af(\iota_O):\Af(\Mb|_O)\longrightarrow \Af(\Mb)$
whose image $\Af^\kin(\Mb;O)$ is called the \emph{kinematic subalgebra} of $\Af(\Mb)$ corresponding to $O$. The correspondence $O\mapsto \Af^\kin(\Mb;O)$  defines
a net of local algebras with properties generalizing those of 
algebraic quantum field theory~\cite{Haag} to curved spacetime~\cite{BrFrVe03}. 

Alternatively, instead of ignoring the background outside $O$ altogether, one could
take the view that it might be altered, and that this 
should have no impact on the physics within $O$.
We employ the relative Cauchy evolution as a quantitative measure 
of the response of such observables to changes in the geometry. For any compact subset $K$ of $\Mb$, set
\[
\Af^\bullet(\Mb;K) := \{A\in \Af(\Mb): \rce_\Mb[\hb]A = A~\text{for all $\hb$ supp in $K^\perp$}\}
\]
where $K^\perp = \Mb\setminus J_\Mb(K)$ is the causal complement of $K$. The \emph{dynamical algebra} $\Af^\dyn(\Mb;O)$ 
associated with an open causally convex region $O$ is then the
subalgebra of $\Af(\Mb)$ generated by the $\Af^\bullet(\Mb;K)$
as $K$ ranges over a suitable set of compact subsets of $O$ (see
 \cite[\S 5]{FewVer:dynloc_theory} for details). 
 
The reader might wonder why only metric perturbations in the
causal complement are considered instead of arbitrary perturbations outside $O$.
The reason for the restriction lies in our use of the relative Cauchy evolution. As mentioned in Section~\ref{sec:rce}, the r.c.e.\ uses the time-slice property to fix an equivalent description of an observable in the causal future of the
perturbation region. It is this initially equivalent description that is held
fixed during the perturbation, rather than the original description
of the observable in $O$, so such observables are not
in general invariant under the relative Cauchy evolution metric perturbations supported in $J_\Mb(O)$. However, the relative
Cauchy evolution induced by perturbations in the causal complement does
provide a usable test of stability under background perturbations. 

If the kinematic and dynamical descriptions agree,
i.e., $\Af^\kin(\Mb;O)=\Af^\dyn(\Mb;O)$ for all open
causally convex subsets $O$ of $\Mb$ with finitely many connected components, the theory is said to be \emph{dynamically local}. This property is known to 
hold for the following quantized theories:
\begin{itemize}
\item the free Klein--Gordon field $(\Box+m^2+\xi R)\phi=0$ in dimensions $n\ge 2$ provided either the mass $m$ or curvature coupling $\xi$ is nonzero~\cite{FewVer:dynloc2,Ferguson:2013}, and 
the corresponding extended algebra of Wick polynomials for
$m>0$ at least for minimal or conformal coupling~\cite{Ferguson:2013}
[there is no reason to expect failure for other values of $\xi$];
\item the free massless current in dimensions $n\ge 2$ (restricting to connected spacetimes) or $n\ge 3$ (allowing disconnected spacetimes)
~\cite{FewVer:dynloc2};
\item the inhomogeneous minimally coupled Klein--Gordon field, 
for $m\ge 0$, $n\ge 2$ -- here one adapts the setting
to $\LocSrc$ by defining the dynamical algebras to be those invariant
under all background perturbations (metric and source) 
in the causal complement~\cite{FewSch:2014};
\item the free Dirac field with mass $m\ge 0$~\cite{Ferguson_PhD};
\item the free Maxwell field in dimension $n=4$, in a 
`reduced formulation' \cite{FewLang_Maxwell:2014}.
\end{itemize}
The known cases in which dynamical locality fails are:
the free Klein--Gordon field with $m=0$, $\xi=0$ in dimensions $n\ge 2$,
which may be traced to the rigid $\phi\mapsto\phi+\text{const}$ gauge symmetry~\cite{FewVer:dynloc2}; the free massless current in $2$-dimensions allowing disconnected spacetimes~\cite{FewVer:dynloc2};
and the free Maxwell field in dimension $n=4$, in a 
`universal formulation' \cite{FewLang_Maxwell:2014}. 
The main difference between the two Maxwell formulations is that 
the universal formulation allows for topological electric and magnetic
charges in spacetimes with nontrivial second de Rham cohomology, whereas the reduced formulation does not. In fact, the topological
charges also fail to satisfy the injectivity property normally required
of a locally covariant theory -- this reflects their nonlocal nature, of course (see \cite{SandDappHack:2012,BecSchSza:2014} for more discussion of the injectivity issue in related models). The emerging pattern is that dynamical locality can
be expected to fail where a theory admits a broken rigid
gauge symmetry, or has charges stabilized by topological or other
constraints; otherwise, it appears to be reasonable to expect
dynamical locality to hold -- it also seems that even trivial
`interactions' such as a mass or curvature coupling are sufficient
to restore dynamical locality. With the exception of the example of the 
massless current in disconnected spacetimes, all known failures of dynamical locality are related to the 
existence of elements that are invariant under arbitrary relative Cauchy evolution;
conceivably the treatment of the massless current might be modified to restore dynamical locality.

\section{The same physics in all spacetimes (SPASs)}

We have seen that the locally covariant approach provides
a criterion for whether two theories represent the same physics
as each other, namely, the existence or otherwise of a natural
isomorphism between the corresponding functors. In
this section, we turn to the question of what can be said about
whether an individual theory is one that represents the same
physics in different spacetimes (see~\cite{FewVer:dynloc_theory}, and~\cite{FewsterRegensburg}
for a summary). 
This is a foundational question for theories of physics in 
curved spacetimes, but one that does not seem to have
been addressed in any axiomatic way before. For theories
defined by a Lagrangian, one may of course write down `the same' Lagrangian in different spacetimes (although there can be
subtleties with this~\cite{FewsterRegensburg}) -- what we
want to understand is on what grounds one might declare that
this is a sound procedure. There are a number of difficulties: 
we lack a definition of `the same physics' and it is unclear 
whether there might be many possible definitions, or indeed whether
any suitable definition exists. We also face the question of
how one should formulate a notion of SPASs mathematically. 

The last question is the easiest to answer: we may represent any possible
definition of SPASs  extensionally by the class of theories that obey it. We then reason as follows:  Suppose $\mathfrak{T}$ is a class of locally covariant theories representing a definition of SPASs,
that $\Af$ and $\Bf$ are theories in $\mathfrak{T}$ and there is a spacetime $\Mb$ in which 
$\Af$ and $\Bf$ represent identical physics. Then, because the physical content of $\Af$ 
is assumed to be the same across all spacetimes, and the same is assumed of $\Bf$ 
(according to a common notion, expressed by $\mathfrak{T}$) the two theories
ought to coincide in every spacetime. The spirit of this argument may be captured 
mathematically as follows: 
\begin{definition} A class of theories $\mathfrak{T}$ has the \emph{SPASs property} if, 
whenever $\Af$ and $\Bf$ are theories in $\mathfrak{T}$, 
such that a natural $\zeta:\Af\nto\Bf$ embeds $\Af$ as a subtheory
of $\Bf$, and there is a spacetime $\Mb$ in which $\zeta_\Mb$ is an isomorphism, then $\zeta$ is a natural isomorphism 
making $\Af$ and $\Bf$ equivalent. 
\end{definition}
Note that the SPASs property is intended as a necessary criterion for $\mathfrak{T}$ to be a satisfactory notion of SPASs. Perhaps surprisingly, the collection of all locally covariant theories, $\LCT$, does not have the SPASs property. Here, we assume $\Phys$ has a monoidal structure understood as composition
of independent systems, so that any $\Af\in\LCT$ can be `doubled' to give 
$\Af^{\otimes 2}\in\LCT$ by 
\[
\Af^{\otimes 2}(\Mb) = \Af(\Mb)\otimes \Af(\Mb),\qquad
\Af^{\otimes 2}(\psi) = \Af(\psi)\otimes \Af(\psi)
\]
for all objects $\Mb$ and morphisms $\psi$ of $\Loc$. 
Let us assume that $\Af$ is inequivalent to $\Af^{\otimes 2}$.\footnote{This holds, for example, if $\Af$ is the theory of a free massive scalar field, because $\Af$ has automorphism group $\ZZ_2$, 
while $\Af^{\otimes 2}$ has automorphism group $\mathrm{O}(2)$~\cite{Fewster:gauge}.} Then we may form a new theory 
$\Bf\in\LCT$ by
\[
\Bf(\Mb) = \begin{cases} \Af(\Mb) & \Sigma_\Mb~\text{noncompact}\\
 \Af(\Mb)^{\otimes 2} & \Sigma_\Mb~\text{compact}\end{cases}~  
\Bf(\psi) A =\begin{cases} \Af(\psi) A & \Sigma_\Nb~\text{noncompact} \\
\Af(\psi)^{\otimes 2} A & \Sigma_\Mb~\text{compact}\\
\Af(\psi)A\otimes\II &  \text{$\Sigma_\Nb$ compact, but not $\Sigma_\Mb$}
\end{cases}
\]
for $\psi:\Mb\to\Nb$, where $\Sigma_\Mb$ denotes a 
Cauchy surface of $\Mb$. Owing to a theorem of Lorentzian geometry~\cite[Thm 1]{BILY} there are no morphisms in which $\Sigma_\Mb$ is compact, but $\Sigma_\Nb$ is noncompact --
indeed, if $\Sigma_\Mb$ is compact then it is diffeomorphic to $\Sigma_\Nb$~\cite[Prop.~A.1]{FewVer:dynloc_theory}.

The theory $\Bf$ represents one copy of $\Af$ in spacetimes with noncompact Cauchy surfaces,
but two copies in spacetimes with compact Cauchy surfaces. It is not hard to show that there are subtheory embeddings $\zeta:\Af\nto\Bf$ and $\eta:\Bf\nto \Af^{\otimes 2}$
\[
\zeta_\Mb A = 
\begin{cases}
A & \Sigma_\Mb~\text{noncompact} \\
A\otimes\II & \Sigma_\Mb~\text{compact}
\end{cases}\qquad
\eta_\Mb A = 
\begin{cases}
A & \Sigma_\Mb~\text{compact} \\
A\otimes\II & \Sigma_\Mb~\text{noncompact}
\end{cases}
\] 
and that $\zeta_\Mb$ is an isomorphism if $\Sigma_\Mb$ is noncompact, while $\eta_\Mb$ is an isomorphism if $\Sigma_\Mb$ is compact.  If $\LCT$ had the SPASs property then there would
be natural isomorphisms $\Af\cong\Bf\cong \Af^{\otimes 2}$,
contradicting the assumed inequivalence of $\Af$ and $\Af^{\otimes 2}$. 
This proves:
\begin{theorem} $\LCT$ does not have the SPASs property,
nor does any class of theories containing $\Af,\Bf,\Af^{\otimes 2}$.
\end{theorem} 
We note that the theory $\Bf$ is just one among many potential
pathological theories that exist in $\LCT$, for which there are
general constructions~\cite[\S 4]{FewVer:dynloc_theory}.

Evidently, in order to find classes of theories that do have the SPASs condition, we need a way of excluding theories like $\Bf$. A clue
is that dynamical locality fails in $\Bf$: if $\Mb$ has compact Cauchy surfaces and $O\subset \Mb$ has nontrivial causal complement, then
\begin{align*}
\Bf^{\kin}(\Mb;O) & = \Af^\kin(\Mb;O)\otimes\II \\[0.1cm]
\Bf^\dyn(\Mb;O ) & = \Af^\dyn(\Mb;O)^{\otimes 2}.
\end{align*}
We see that $\Bf^\dyn(\Mb;O)$ captures the degrees of freedom available in the ambient spacetime, owing to the
use of the relative Cauchy evolution in its definition. This
clue turns out to be fruitful. The main result of the general 
analysis (\cite[Thm 6.10]{FewVer:dynloc_theory}) is 
\begin{theorem}\label{thm:SPASs}
The class of dynamically local theories has the SPASs property.  
\end{theorem}
Aside from its intrinsic interest, this result has an application to the question
of the existence of natural states~\cite[Thm 6.13]{FewVer:dynloc_theory} (we state a simplified and slightly weaker version):
\begin{theorem} \label{thm:nogo}
Suppose $\Af:\Loc\to\Alg$ is a dynamically local theory that also obeys extended locality\footnote{
Extended locality~\cite{Schoch1968, Landau1969} requires that the kinematic algebras of spacelike separated regions
intersect only in multiples of the unit. The Reeh--Schlieder property 
holds if the GNS vector corresponding to $\omega_\Mb$ is cyclic for the induced representation of 
$\Af(\Mb|_O)$ for all open, relatively compact, connected causally convex $O$ -- it 
is a standard feature of QFT in Minkowski space~\cite{StreaterWightman,Haag}
 and has also been proved for the free scalar field in some curved spacetime situations~\cite{StrVerWol:2002}.} and admits a natural state $(\omega_\Mb)_{\Mb\in\Loc}$. 
If, in Minkowski space $\Mb_0$, the state $\omega_{\Mb_0}$ induces a 
faithful GNS representation of $\Af(\Mb_0)$ with the Reeh--Schlieder property,  then $\Af$ is equivalent to the trivial theory 
that assigns the trivial unital $*$-algebra $\CC$ to every spacetime. 
\end{theorem}
Sketch arguments
for nonexistence of natural states of the real scalar field appear in \cite{Ho&Wa01,BrFrVe03};
however, Theorem~\ref{thm:nogo} was the first complete argument and, moreover, 
applies to a general class of theories, including those listed in Section~\ref{sec:dynloc}.  Its proof makes use of Theorem~\ref{thm:SPASs}
-- the trivial theory is certainly a subtheory of $\Af$,  and the hypotheses entail that these theories coincide in Minkowski space; hence, 
as both $\Af$ and the trivial theory are dynamically local, the result is proved.

\section{Concluding remarks}

This paper has described the framework of locally covariant physical theories, providing
a motivation in terms of `ignorance principles'. This framework has opened up
axiomatic QFT in curved spacetime: in addition to the results discussed above, there is a general spin-statistics
theorem~\cite{Verch01} and versions of the Reeh--Schlieder theorem~\cite{Sanders_ReehSchlieder}, Haag duality~\cite{Ruzzi_punc:2005} and the split property~\cite{Few_split:2015},  
the global gauge group is understood~\cite{Fewster:gauge} and superselection sectors
have been investigated~\cite{Br&Ru05}. Moreover, the underlying ideas play an important role
in the perturbative construction of interacting models in curved spacetime~\cite{BrFr2000,Ho&Wa01}
among other applications. Although the formalism allows us to begin to address the issue of SPASs,
more can be done in this direction: while the class of dynamically local theories 
has the SPASs property, it is unknown whether it may contain pathological theories that
should be ruled out by further axioms.

As far as natural states are concerned, the general result of Theorem~\ref{thm:nogo} 
shows that they cannot be expected in reasonable models of QFT; 
we have also given a relatively straightforward proof of this for the inhomogeneous scalar field.
As with any no-go theorem, one can always seek to by-pass the hypotheses. The obvious condition to
drop is the requirement that the preferred state depend locally on the background, and indeed a proposal
of this type has been discussed recently for the scalar field~\cite{AAS}. However,  the states constructed
turn out not to be Hadamard and to have a number of other defects~\cite{CF-RV-ultraSJ,CF-RV-morebadnews} (see also~\cite{FewLan_Dirac:2014} for related discussion for Dirac fields).
While there is an ingenious modification of~\cite{AAS} that does yield
Hadamard states --- see ~\cite{BruFre:2014} (and ~\cite{FewLan_Dirac:2014} in the Dirac case) 
this is achieved at the expense of introducing a whole 
family of states, none of which is canonically preferred. 
The no-go theorem is not so easily evaded.

\medskip

\noindent{\small I thank the organizers and participants of the workshop
\emph{New Geometric Concepts in the Foundations of Quantum Physics}
(Chicheley Hall, November 2013) for stimulating discussion and comments, and also the Royal Society for  funding the workshop. I also thank Francis Wingham and the referees for their careful reading of the manuscript.}

{\small
\providecommand{\newblock}{}
}
%

\end{document}